\documentclass[A4]{amsart}
\usepackage{amsthm,amsmath,amssymb,eucal}
\usepackage{geometry}
\usepackage{resizegather}
\usepackage{graphicx}
\usepackage{color}
\usepackage[normalem]{ulem}
\usepackage[latin1]{inputenc}

\newtheorem{claim}{Claim}[section]
\newtheorem{theorem}[claim]{Theorem}

\newtheorem{lemma}[claim]{Lemma}

\newtheorem{corollary}[claim]{Corollary}

\newtheorem{thmx}{Theorem}

\newcommand{\soutg}{\bgroup\markoverwith{\textcolor{green}{\rule[.5ex]{2pt}{1pt}}}\ULon}
\newcommand{\soutb}{\bgroup\markoverwith{\textcolor{blue}{\rule[.5ex]{2pt}{1pt}}}\ULon}
\newcommand{\soutr}{\bgroup\markoverwith{\textcolor{red}{\rule[.5ex]{2pt}{1pt}}}\ULon}

\newcommand{\bo}{{\rm O}}
\newcommand{\so}{{\rm o}}
\newcommand{\ds}{\displaystyle}
\newcommand{\dsum}{\ds\sum}

\newcommand{\doint}{\ds\oint}
\newcommand{\eqskip}{ \vspace*{2mm}\\ }
\newcommand{\fr}[2]{\frac{\ds #1}{\ds #2}}

\title[Gelfand-Levitan trace formula on quantum graphs]{A Gelfand-Levitan trace formula for generic quantum graphs}

\author{Pedro Freitas} 
\author{Ji\v{r}\'{\i} Lipovsk\'{y}}

\address{Departamento de Matem\'atica, Instituto Superior T\'ecnico, Universidade de Lisboa, Av. Rovisco Pais 1,
P-1049-001 Lisboa, Portugal {\rm and}
Grupo de F\'{\i}sica Matem\'{a}tica, Faculdade de Ci\^encias, Universidade de Lisboa,
Campo Grande, Edif\'icio C6, P-1749-016 Lisboa, Portugal}
\email{psfreitas@fc.ul.pt}
\address{Department of Physics, Faculty of Science, University of Hradec Kr\'alov\'e, Rokitansk\'eho 62,
500\,03 Hradec Kr\'alov\'e, Czechia}
\email{jiri.lipovsky@uhk.cz}

\begin{document}

\begin{abstract}
We formulate and prove a Gelfand-Levitan trace formula for general quantum graphs with arbitrary edge lengths and 
coupling conditions which cover all self-adjoint operators on quantum graphs, except for a set of measure zero.
The formula is reminiscent of the original Gelfand-Levitan result on the segment with Neumann boundary conditions.
\end{abstract}

\maketitle

\section{Introduction}
Given a Schr\"odinger operator with a potential $q$ on a line segment of length $\pi$ with Neumann boundary conditions,
let us denote the corresponding eigenvalues by $\lambda_n(q)$. In~\cite{GL1} Gelfand and Levitan found and proved a formula for the
sum of the differences between $\lambda_n(q)$ and the eigenvalues of the null potential $\lambda_n(0)$, namely,
\begin{equation}\label{gelflevi}
 \sum_{n=1}^\infty \left[ \lambda_n(q)-\lambda_n(0) -\fr{1}{\pi}\int_0^\pi q(x)\,\mathrm{d}x\right] =
  \fr{1}{4}[q(\pi)+q(0)] - \fr{1}{2\pi}\int_0^\pi q(x)\,\mathrm{d}x\,,
\end{equation}
under certain regularity conditions -- see also~\cite{diki,HK1}. Since then,
regularised trace formulas of this type have been present in the literature more or less continuously and were extended to many 
different settings and forms, including more general operators and potentials -- see~\cite{sapo} for a review of the topic, including 
some historical notes. Also, and as was pointed out in~\cite{B}, there is a relation between the trace formula~\eqref{gelflevi} and the 
short-time asymptotic expansion of the trace of the heat kernel.

Of interest to us here are the extensions to the case of quantum graphs, where this type of
result may be traced back to the papers by Roth~\cite{Ro}, and Kottos and Smilansky~\cite{KS99}, with further developments in several directions
such as those in~\cite{BER,BK,frke,N}.
Quantum graphs have also received much attention in the literature within the past 30 years and, in particular, there have been several
attempts at generalizing Gelfand and Levitan's result to this setting. So far, the results obtained have been restricted to specific graphs and
include, for instance, the case of equilateral graphs for which Carlson proved a formula involving integrals of the potential and the
eigenfunctions~\cite{carl},  and the work of C--F. Yang and J.-X. Yang for equilateral star graphs with different boundary conditions
and coupling at the central node, which are closer in form to~\eqref{gelflevi}~\cite{YY,Yan13}.

A main difficulty with extending~\eqref{gelflevi} to graphs with a general topology and arbitrary edge lengths is that there
will then exist eigenvalue sequences with different asymptotic behaviours, making the regularisation of the trace by associating
the different eigenvalues of the problem with a potential to those with the null potential a delicate issue.
The purpose of the present paper is to provide an answer to this question in this general setting. We thus consider graphs
with arbitrary edge lengths and topology, while the coupling is generic in the following sense. The whole class of coupling 
conditions defining a self-adjoint operator is allowed, with the exception of a set of measure zero corresponding to a particular
eigenvalue in the coupling matrix. This exception leaves out some important coupling conditions such as Dirichlet, standard or
$\delta$-coupling, but it does include Robin, Neumann or $\delta'$-coupling (the last one with the exception of the case
when the coupling parameter is zero). However, notice that this distinction is natural, as there are indeed differences between
these two families of coupling conditions and thus the corresponding trace formulas are expected to differ in the two cases. This is already
visible in the case of a single interval where, for instance, the original Gelfand-Levitan formula with Dirichlet boundary
conditions \cite{GL1,HK1} reads as
$$
  \sum_{n=1}^\infty [\lambda_n(q)-\lambda_n(0)]	 = \frac{1}{4}[q(\pi)+q(0)]\,.
$$
Furthermore, the local scattering matrices for the set of coupling conditions we consider converge for high energies to the matrices for decoupled
Neumann conditions (for details, see e.g. \cite{BE,BK}).

A key point in our approach is that, unlike in~\cite{YY, Yan13}, for instance, we do not subtract from the eigenvalue of the 
Hamiltonian with the potential a particular known value -- in~\eqref{gelflevi}, $\lambda_{n}(0)$ is in fact $k^2$ --, but
for elegance of the result, we find the formula for the difference between the eigenvalues of the Hamiltonian with the potential
and those for when this potential is zero. This is, in fact, what allows us to assign a correspondence between the different
eigenvalues in such a way as to make the involved series convergent, while not making the corresponding formula cumbersome.
If, for instance, one considers the first terms in the asymptotics of the different sequences, then the formula will include
terms related to the coupling matrix, for instance. 

The discussion above also implies that a second ingredient which is necessary to obtain for a formula of this type to work is
the asymptotic behaviour of the different sequences up to an order such that we can both group eigenvalues according to
their asymptotic behaviour and ensure convergence of the series involved. As far as we are aware, previous results along
these lines for general graphs only considered remainders of order zero~\cite{N}, while for our purposes we need to go up to the
term with remainder of order $n^{-2}$.

The separation of the full spectrum into different sequences may be done in several different ways, so we now
briefly explain our procedure. It is clear that, with the end in view, the basis for this separation has to be the asymptotic
behaviour of the spectrum. We first note that the leading term of the secular equation is $\prod_{i = 1}^d(-k\sin{(k\ell_i)})$,
where $d$ is the number of edges in the graph, $k$ is the square root of the energy and $\ell_i$ are the edge lengths.
In Section~\ref{sec:main} we prove that the square roots of eigenvalues are close to the zeros of the given product and that they
can be grouped in sets of at most $d$ eigenvalues and $d$ zeros. We thus partition the spectrum into $d$ subsequences of eigenvalues
in the following way. Denote the sequence of all eigenvalues in increasing order by $\{\lambda_n\}_{n=1}^\infty$ and let the sequence
$\{\mu_n\}_{n=1}^\infty$ correspond to the non-negative zeros of the above product, also arranged in increasing order, with the first
$d$ entries being $0$. We now pair $\lambda_n$ with $\mu_n$ and define the subsequences $\{\lambda_{in}\}_{n=0}^\infty$ as subsequences of
$\{\lambda_n\}_{n=1}^\infty$ which are paired with those zeros of $\prod_{i = 1}^d(-k\sin{(k\ell_i)})$ which are zeros of $\sin{(k\ell_i)}$ 
for a given $i$ (the first entry of this sequence $\lambda_{i0}$ is paired
with 0).

We may now formulate the main result of the paper.
\begin{thmx}\label{thm:main}
We assume a quantum graph with $d$ edges with arbitrary lengths $\ell_{i}$, $i=1,\dots,d$, and associated coupling matrix $U$ not
having $-1$ in its spectrum. Then, denoting the eigenvalues of the Hamiltonian with a potential $q$ and with the zero potential  by
$\lambda_{in}(q)$ and $\lambda_{in}(0)$, respectively, in the way described above, and the component of the potential on the $i$-th edge by $q_i 
\in W^{1,1}((0,\ell_i))$, the following trace formula holds 
$$
  \sum_{i = 1}^d \sum_{n = 0}^\infty \left[\lambda_{in}(q) - \lambda_{in}(0) - \frac{1}{\ell_i}\int_0^{\ell_i} q_i(x)\,\mathrm{d}x\right] = 
 \sum_{i = 1}^d \left\{\frac{1}{4}\left[q_i(\ell_i)+q_i(0)\right]-\frac{1}{2\ell_i}\int_0^{\ell_i} q_i(x)\,\mathrm{d}x \right\}\,.
$$
\end{thmx}

The proof of Theorem~\ref{thm:main} proceeds along the following lines. We first derive the corresponding secular equation,
which we compare with the product of sine functions with arguments associated to the lengths of each edge. This allows us to divide the eigenvalues
into at most $d$ groups as described above, and derive the corresponding asymptotic behaviour of each of these sequences, yielding
the approximate location of eigenvalues with the necessary accuracy. With this, we may then prove the absolute convergence of the sum on the left-hand side
of the trace formula, needed to justify the necessary rearranging of the summands. 
The main techniques used in the proof rely on the complex integration of functions related to the secular equation along appropriately chosen contours.
These correspond to the boundaries of a sequence of growing embedded squares, for which Rouch\'{e}'s theorem allows us to obtain that the number of zeros of
a modified form of the secular equation and the corresponding product of sine functions mentioned above are the same.

The paper is structured as follows. In the next section we describe the model of quantum graphs, and in Section~\ref{sec:seceq}
the secular equation is found and some preparatory calculations for Section~\ref{sec:main} are performed. In Section~\ref{sec:main} we give
the proofs of the main results leading to the proof of Theorem~\ref{thm:main}. Several technical results used throughout the paper
are given in the appendices.

\section{Description of the model}\label{sec:description}
We briefly introduce the model of quantum graphs; for more details we refer the reader to~\cite{BK}. Let us consider a metric graph $\Gamma$
consisting of the set of vertices $v\in \mathcal{V}$, which are connected by the set of $d$ finite edges $e_j\in \mathcal{E}$. The number $d$
is finite and the lengths of the edges are $\ell_j\in(0,\infty)$. We equip the graph $\Gamma$ with a self-adjoint operator
$$
  \mathcal{H} = -\frac{\mathrm{d}^2}{\mathrm{d}x^2}+q_j(x)\,,\quad x\in e_j
$$
with the real potentials $q_j\in W^{1,1}(e_j)$. The domain of $\mathcal{H}$ consists of functions with the edge components in the Sobolev
spaces $W^{2,2}(e_j)$ and satisfying the coupling conditions
$$
  (U_v-I)\Psi_v+i(U_v+I)\Psi_v'=0
$$
at the vertices. Here $U_v$ is a $d_v\times d_v$ unitary matrix ($d_v$ is the degree of the vertex $v$), $\Psi_v$ is the vector of the limiting
values of functions at the vertex $v$ from its incident edge and, similarly, $\Psi_v'$ is the vector of the derivatives outgoing from $v$; $I$ is the
$d_v\times d_v$ identity matrix. Throughout the paper, we will assume that $-1\not \in \sigma(U_v)$.

With the use of the flower-like model (see \cite{Ku3,EL2}), where all the vertices are joined into one and the topology of the graph is described by the 
larger $2d\times 2d$ coupling matrix $U$, one may write the coupling condition as
$$
  (U-I)\Psi +  i (U+I)\Psi' = 0\,.
$$
Here, with a small abuse of notation, $I$ refers now to the $2d\times 2d$ identity matrix, $\Psi$ is the vector with the limiting values of 
functions defined on each edge, as the vertex is approached from either end of the edge, and
$\Psi'$ is the vector of limits of the corresponding outgoing derivatives. We assume that the first entry of the vector $\Psi$ is the functional value at the beginning
of the first edge, the second entry is the functional value at the end of the first edge, the third entry is the functional value at the beginning of the second edge, and so on, and similarly for $\Psi'$. Using the fact that $-1$ is not in $\sigma(U)$, we may write
\begin{equation}
  H \Psi + \Psi' = 0\,,\label{eq:coupling}
\end{equation}
where $H = -i(U+I)^{-1}(U-I)$ is a Hermitian $2d\times 2d$ matrix. We denote the entries of the matrix $H$ in the following way
\begin{eqnarray}\label{matrixH}
  H_{2i-1,2i-1}=: H_{11i}\,,\quad H_{2i-1,2i}=: H_{12i}\,,\quad H_{2i,2i}=:H_{22i}\,,\nonumber \\
  H_{2i-1,2j-1}=: H_{11ij}\,,\quad H_{2i-1,2j}=: H_{12ij}\,,\\
  H_{2i,2j-1}=: H_{21ij}\,,\quad H_{2i,2j}=: H_{22ij}\,,\quad i<j\,. \nonumber
\end{eqnarray}

\section{The secular equation}\label{sec:seceq}
In this section, we will construct the secular equation. On each edge, we consider the two independent solutions $c_j$ and $s_j$ of the initial value problem
$\mathcal{H}u(x,k) = k^2 u(x,k)$, satisfying the conditions $c_j(0,k) = 1$, $c_j'(0,k) = 0$, 
and $s_j(0,k) = 0$, $s_j'(0,k) = 1$, which are then given by
\begin{eqnarray*}
  c_j(x,k) = \cos{(kx)}+\int_0^x \frac{\sin{(k(x-t))}}{k}q_j(t)c_j(t,k)\,\mathrm{d}t\,,\\
  s_j(x,k) = \frac{\sin{(kx)}}{k}+\int_0^x \frac{\sin{(k(x-t))}}{k}q_j(t)s_j(t,k)\,\mathrm{d}t\,.
\end{eqnarray*}

The following lemma (the asymptotic expansion follows the idea of \cite{Yur}) is proven in Appendix~\ref{sec:appendixa}. 
\begin{lemma}\label{lem:cs}
The functions $c_{j}$ and $s_{j}$ defined above satisfy
\begin{eqnarray*}
  c_j(\ell_j,k) &=& \cos{(k\ell_j)}+a_j\frac{\sin{(k\ell_j)}}{k} +\so\left(\frac{\mathrm{e}^{|\mathrm{Im\,k}|\ell_j}}{k}\right)\,,\\
  c_j'(\ell_j,k) &=& -k\sin{(k\ell_j)}+a_j\cos{(k\ell_j)} + b_j \frac{\sin{(k\ell_j)}}{k}+\so\left(\frac{\mathrm{e}^{|\mathrm{Im\,k}|\ell_j}}{k}\right)\,,\\
  s_j(\ell_j,k) &=& \frac{\sin{(k\ell_j)}}{k}-a_j\frac{\cos{(k\ell_j)}}{k^2}+ \so\left(\frac{\mathrm{e}^{|\mathrm{Im\,k}|\ell_j}}{k^2}\right)\,,\\
  s_j'(\ell_j,k) &=&\cos{(k\ell_j)}+a_j\frac{\sin{(k\ell_j)}}{k}+\so\left(\frac{\mathrm{e}^{|\mathrm{Im\,k}|\ell_j}}{k}\right)
\end{eqnarray*}
with 
$$
  a_j := \frac{1}{2}\int_0^{\ell_j}q_j(t)\,\mathrm{d}t\,,\quad b_j := \frac{q_j(\ell_j)+q_j(0)}{4}+\frac{1}{8}\left(\int_0^{\ell_j}q_j(t)\,\mathrm{d}t\right)^2\,.
$$

\end{lemma}

We will now transform equation~\eqref{eq:coupling} into a form that is more appropriate for our purposes.
Using the expression of the components of the eigenfunction as the linear combination
$$
  f_j(x) = A_j c_j(x,k)+B_j s_j(x,k)\,,
$$
and the corresponding initial conditions for $c_j$ and $s_j$, we may rewrite equation~\eqref{eq:coupling} as
$$
  [H M_1(k)+M_2(k)](A_1,B_1,A_2,B_2,\dots, A_d,B_d)^\mathrm{T} = 0\,.
$$
Here the matrices $M_{1}$ and $M_{2}$ are given by 
\begin{gather*}
M_1(k) = \begin{pmatrix}1& 0 & 0 & 0 & \dots\\ c_1(\ell_1,k) & s_1(\ell_1,k) & 0 & 0 & \dots\\ 0 & 0 & 1 & 0 & \dots\\ 0 & 0 & c_2(\ell_2,k) & s_2(\ell_2,k) & \dots\\ \vdots &  \vdots &  \vdots &  \vdots & \ddots\end{pmatrix}
\\
  = \begin{pmatrix}1& 0 & 0 & 0 & \dots\\ \cos{(k\ell_1)}+a_1 \fr{\sin{(k\ell_1)}}{k} +r_1 & \fr{\sin{(k\ell_1)}}{k}-a_1\fr{\cos{(k\ell_1)}}{k^2}+r_2  & 0 & 0 & \dots\\ 0 & 0 & 1 & 0 & \dots\\ 0 & 0 & \cos{(k\ell_2)}+a_2\fr{\sin{(k\ell_2)}}{k} +r_1 & \fr{\sin{(k\ell_2)}}{k}-a_2\fr{(\cos{k\ell_2)}}{k^2}  +r_2 & \dots\\ \vdots &  \vdots &  \vdots &  \vdots & \ddots \end{pmatrix}\,,
\end{gather*}
where $r_1 = \so\left(\fr{\mathrm{e}^{|\mathrm{Im\,k}|\max\ell_j}}{k}\right)$ and $r_2 = \so\left(\fr{\mathrm{e}^{|\mathrm{Im\,k}|\max\ell_j}}{k^2}\right)$\,,
and
\begin{gather*}
  M_2(k) = \begin{pmatrix}0& 1 & 0 & 0 & \dots\\ -c_1'(\ell_1,k) & -s_1'(\ell_1,k) & 0 & 0 & \dots\\ 0 & 0 & 0& 1 & \dots\\ 0 & 0 & -c_2'(\ell_2,k) & -s_2'(\ell_2,k) & \dots\\ \vdots &  \vdots &  \vdots &  \vdots & \ddots\end{pmatrix}
\\
 = {\begin{pmatrix}0& 1 & 0 & 0 & \dots\\ k \sin{(k\ell_1)} - a_1\cos{(k\ell_1)}-b_1\fr{\sin{(k\ell_1)}}{k} & -\cos{(k\ell_1)}-a_1\fr{\sin{(k\ell_1)}}{k} & 0 & 0 & \dots\\ 0 & 0 & 0 & 1 & \dots\\ 0 & 0 & k \sin{(k\ell_2)} - a_2\cos{(k\ell_2)}-b_2\fr{\sin{(k\ell_2)}}{k} & -\cos{(k\ell_2)}-a_2\fr{\sin{(k\ell_2)}}{k}  & \dots\\ \vdots &  \vdots &  \vdots &  \vdots & \ddots \end{pmatrix}}+
\\ + \so\left(\fr{\mathrm{e}^{|\mathrm{Im\,k}|\max\ell_j}}{k}\right)\,,
\end{gather*}

Hence the secular equation may be written as $\varphi(k):=\ \mathrm{det\,}[H M_1(k)+M_2(k)] = 0$. From this we obtain, after a straightforward, but
rather tedious computation
\begin{multline*}
  0 = \varphi(k) = \prod_{i=1}^d (-k\sin{(k\ell_i)})+\sum_{i = 1}^d \left(\prod_{\stackrel{j = 1}{j\ne i}}^d (-k\sin{(k\ell_j)})\right)[\cos{(k\ell_i)}(a_i-\mathrm{Tr\,}H_i)-2\mathrm{Re\,}H_{12i}]+
\\
 + \sum_{\stackrel{i,j=1}{i<j}}^d\left(\prod_{\stackrel{o = 1}{\so\ne i,j}}^d(-k\sin{(k\ell_o)})\right)\left\{\frac{\sin{(k\ell_i)}\sin{(k\ell_j)}}{d-1}(a_i \mathrm{Tr\,}H_i+a_j\mathrm{Tr\,}H_j-b_i-b_j-\right.
\\
\left.-\mathrm{det\,}H_i-\mathrm{det\,}H_j)+\cos{(k\ell_i)}\cos{(k\ell_j)}[a_i a_j-a_i\mathrm{Tr\,}H_j-a_j\mathrm{Tr\,}H_i-\right.
\\
\left.-(|H_{11ij}|^2+|H_{12ij}|^2+|H_{21ij}|^2+|H_{22ij}|^2)+\mathrm{Tr\,}H_i\mathrm{Tr\,}H_j]+\right.
\\
 \left.+\cos{(k\ell_i)}[2(\mathrm{Tr\,}H_i-a_i)\mathrm{Re\,}H_{12j}-2\mathrm{Re\,}(H_{11ij}\bar{H}_{12ij}+H_{22ij}\bar{H}_{21ij})]+\right.
\\
 \left.+\cos{(k\ell_j)}[2(\mathrm{Tr\,}H_j-a_j)\mathrm{Re\,}H_{12i}-2\mathrm{Re\,}(H_{11ij}\bar{H}_{21ij}+H_{22ij}\bar{H}_{12ij})]+\right.
\\
  +4\mathrm{Re\,}H_{12i}\mathrm{Re\,}H_{12j}-2\mathrm{Re\,}(H_{12ij}\bar{H}_{21ij}+H_{11ij}\bar{H}_{22ij})\bigg\}+\so\left(k^{d-2}\mathrm{e}^{|\mathrm{Im\,k}|\sum_{i=1}^d\ell_i}\right)\,.
\end{multline*}
where $H_i = \begin{pmatrix}H_{11i} & H_{12i}\\ \bar{H}_{12i} & H_{22i}\end{pmatrix}$ -- see~\eqref{matrixH} for the definition of the 
entries of the matrix $H$.

Dividing the above formula by $\prod_{i = 1}^d (-k\sin{(k\ell_i)})$, we write the residual term as $\so\left(\frac{1}{k^2}\right)$. Although this does not hold close to the zeros of $\sin{(k\ell_i)}$, it does hold on the contours $\Gamma_N$ and $C_p$ defined below, thus allowing us to compute the integrals on these contours. Using Lemma~\ref{lem:smallo} we find that
\begin{multline*}
  \frac{\varphi(k)}{\prod_{i = 1}^d (-k\sin{(k\ell_i)})} = 1+ \frac{1}{k}\sum_{i = 1}^d \left[\cot{(k\ell_i)}(\mathrm{Tr\,}H_i-a_i)+\frac{2\mathrm{Re\,}H_{12i}}{\sin{(k\ell_i)}}\right]+
\\
 + \frac{1}{k^2}\Bigg\{ \sum_{i=1}^d (a_i \mathrm{Tr\,}H_i-b_i-\mathrm{det\,}H_i)+\sum_{\stackrel{i,j=1}{i<j}}^d\cot{(k\ell_i)}\cot{(k\ell_j)}[a_i a_j -a_i \mathrm{Tr\,}H_j
\\
 \left.-a_j\mathrm{Tr\,}H_i-(|H_{11ij}|^2+|H_{12ij}|^2+|H_{21ij}|^2+|H_{22ij}|^2)+\mathrm{Tr\,}H_i\mathrm{Tr\,}H_j]+\right.
\\
 \left.+\frac{\cot{(k\ell_i)}}{\sin{(k\ell_j)}}[2\mathrm{Tr\,}H_i\mathrm{Re\,}H_{12j}-2a_i\mathrm{Re\,}H_{12j}-2\mathrm{Re\,}(H_{11ij}\bar{H}_{12ij}+H_{22ij}\bar{H}_{21ij})]+\right.
\\
 \left.\frac{\cot{(k\ell_j)}}{\sin{(k\ell_i)}}[2\mathrm{Tr\,}H_j\mathrm{Re\,}H_{12i}-2a_j\mathrm{Re\,}H_{12i}-2\mathrm{Re\,}(H_{11ij}\bar{H}_{21ij}+H_{22ij}\bar{H}_{12ij})]+\right.
\\
  \frac{1}{\sin{(k\ell_i)}\sin{(k\ell_j)}}[4\mathrm{Re\,}H_{12i}\mathrm{Re\,}H_{12j}-2\mathrm{Re\,}(H_{12ij}\bar{H}_{21ij}+H_{11ij}\bar{H}_{22ij})]\Bigg\}+\so\left(\frac{1}{k^2}\right)\,.
\end{multline*}

Using the Taylor expansion for the logarithm around one we obtain
\begin{multline}
  \ln\frac{\varphi(k)}{\prod_{i = 1}^d (-k\sin{(k\ell_i)})} = \frac{1}{k}\sum_{i = 1}^d \left[\cot{(k\ell_i)}(\mathrm{Tr\,}H_i-a_i)+\frac{2\mathrm{Re\,}H_{12i}}{\sin{(k\ell_i)}}\right]+
\\
 \frac{1}{k^2}\sum_{\stackrel{i,j=1}{i<j}}^d \left\{\cot{(k\ell_i)}\cot{(k\ell_j)}[-(|H_{11ij}|^2+|H_{12ij}|^2+|H_{21ij}|^2+|H_{22ij}|^2)]+\right.
\\
 \left.+ \frac{\cot{(k\ell_i)}}{\sin{(k\ell_j)}}[-2\mathrm{Re\,}(H_{11ij}\bar{H}_{12ij}+H_{22ij}\bar{H}_{21ij})]+\right.
\\
 \left.+\frac{\cot{(k\ell_j)}}{\sin{(k\ell_i)}}[-2\mathrm{Re\,}(H_{11ij}\bar{H}_{21ij}+H_{22ij}\bar{H}_{12ij})]+\right.
\\
 \left.\frac{1}{\sin{(k\ell_i)}\sin{(k\ell_j)}}[-2\mathrm{Re\,}(H_{12ij}\bar{H}_{21ij}+H_{11ij}\bar{H}_{22ij})] \right\}
\\
  +\frac{1}{k^2} \sum_{i = 1}^d\left\{\cot^2{(k\ell_i)}\left[-\frac{1}{2}(\mathrm{Tr\,}H_i-a_i)^2\right]+(a_i\mathrm{Tr\,}H_i-b_i-\mathrm{det\,}H_i)\right.
\\
  \left.+\frac{\cot{(k\ell_i)}}{\sin{(k\ell_i)}}[-2(\mathrm{Tr\,}H_i-a_i)\mathrm{Re\,}H_{12i}]-2\frac{(\mathrm{Re\,}H_{12i})^2}{\sin^2{(k\ell_i)}}\right\}+\so\left(\frac{1}{k^2}\right)\,.\label{eq:lnphi/prod}
\end{multline}

Writing $\varphi_0(k)$ for the function in the secular equation when $q_j(x) = 0, j=1,\dots,d$ we obtain in a similar way
\begin{multline}
  \ln\fr{\varphi(k)}{\varphi_0(k)} = -\fr{1}{k}\dsum_{i = 1}^d \cot{(k\ell_i)}a_i +\fr{1}{k^2}\dsum_{i = 1}^d \left[\fr{1}{\sin^2{(k\ell_i)}} a_i \mathrm{Tr\,}H_i-b_i\phantom{\fr{\cot{(k\ell_i)}}{\sin{(k\ell_i)}} }\right.
\\
 \left.-\fr{1}{2}\cot^2{(k\ell_i)}a_i^2+\fr{\cot{(k\ell_i)}}{\sin{(k\ell_i)}}2a_i\mathrm{Re\,}H_{12i}\right]+ \so\left(\fr{1}{k^2}\right) \,.\label{eq:lnphiphi0}
\end{multline}

\section{Proof of the main result}\label{sec:main}
Let us define the counter-clockwise contour $\Gamma_N$ in the complex variable $k$ as a square with vertices $N-iN$, $N+iN$, $-N+iN$, $-N-iN$. Then, using the symmetric
version of Rouch\'e's theorem, we can prove the following theorem relating the number of zeros of $\prod_{i=1}^d (-k\sin{(k\ell_i)})$ and zeros of $\varphi(k)$
(the proof is given in Appendix~\ref{sec:appendixb}).
\begin{theorem}\label{thm:numberofzeros}
For all $\varepsilon>0$ there exists $K>0$ so that for all $N>K$ and $N\not \in \cup_{i=1}^d\cup_{n\in \mathbb{N}_0} \left(\fr{n\pi}{\ell_i}-\fr{\varepsilon}{\ell_i}, \fr{n\pi}{\ell_i}+\fr{\varepsilon}{\ell_i}\right)$ the functions $\prod_{i=1}^d (-k\sin{(k\ell_i)})$ and $\varphi(k)$have the same number of zeros inside the contour $\Gamma_N$.
\end{theorem}

Let us denote the sequence of all eigenvalues of the operator $\mathcal{H}$ arranged by ascending order by $\{\lambda_n\}_{n=1}^\infty$. We denote by $\{\mu_n\}_{n=1}^\infty$ the sequence in which the first $d$ elements are $0$ and all subsequent elements are positive zeros of $\prod_{i = 1}^d\sin{(k\ell_i)}$ arranged in increasing order. We pair $\lambda_n$ with $\mu_n$. In view of Theorem~\ref{thm:numberofzeros}, $k_n : =\sqrt{\lambda_n}$ with
$\mathrm{Re\,}k_n\geq 0$ is ``close to'' $\mu_n$, as we will see in the following lemma. We will denote the sequence of eigenvalues corresponding to the zeros of $\sin{(k\ell_i)}$ by $\{\lambda_{in}\}_{n=0}^\infty$, where $\lambda_{i0}$ corresponds to $0$ and the remaining values to positive zeros of $\sin{(k\ell_i)}$. 

\begin{lemma}
It is possible to choose $\varepsilon>0$ and $K>0$ such that there exists a strictly increasing sequence $\{N_p\}_{p=1}^\infty$ with 
$K<N_1$ and satisfying
\[
 N_{p}\not\in\cup_{i=1}^d\cup_{n\in \mathbb{N}_0} \left(\fr{n\pi}{\ell_i}-\fr{\varepsilon}{\ell_i}, \fr{n\pi}{\ell_i}+\fr{\varepsilon}{\ell_i}\right)
\]
and
\[
 {\ds \lim_{p\to\infty}} N_{p} = +\infty.
\]
and there are at most $d$ eigenvalues $\lambda = k^2$ of $\mathcal{H}$ with $N_p\leq k \leq N_{p+1}$,
for all $p\in \mathbb{N}$. Furthermore, all these eigenvalues belong to different sequences $\lambda_{in}$ and there are at most $d$ zeros $\mu$ of
$\prod_{i = 1}^d\sin{(k\ell_i)}$ with $N_p\leq \mu \leq N_{p+1}$, $\forall p\in \mathbb{N}$. The number of eigenvalues and zeros with this property is the same.
\end{lemma}
\begin{proof}
We choose \[\varepsilon < \fr{\pi}{4\max_{j}\ell_j\dsum_{i=1}^d \fr{1}{\ell_i}}.\] The width of each interval $\left(\fr{n\pi}{\ell_i}-\fr{\varepsilon}{\ell_i}, \fr{n\pi}{\ell_i}+\fr{\varepsilon}{\ell_i}\right)$ is $\fr{2\varepsilon}{\ell_i}$ and so the sum of the lengths of these ``forbidden'' intervals for all sequences is $2\varepsilon \dsum_{i=1}^d \fr{1}{\ell_i}$. We choose $\varepsilon$ sufficiently small to ensure that this expression is smaller than
$\fr{\pi}{2\max_j \ell_j}$ (half of the smallest distance between two neighbouring zeros of the sine function from the given sequence). Hence the ``forbidden intervals'' do not cover the whole interval between two neighbouring zeros of a given sine function, and it is possible to choose a contour 
in Theorem~\ref{thm:numberofzeros} between them and obtain that the number of zeros of the sine and the eigenvalues in that contour is the same. 
\end{proof}

Now we choose for contours $C_p$ the rectangles with vertices $N_{p+1}-i N_{p+1}$, $N_{p+1}+iN_{p+1}$, $N_{p}+i N_{p+1}$ and $N_{p}-i N_{p+1}$,
traversed counter-clockwise. Inside the contour there is the same number of square roots of eigenvalues of $\mathcal{H}$ and zeros of $\prod_{i = 1}^d k\sin{(k\ell_i)}$ and this number is at most $d$. Let us first consider the case when there is only one square root of eigenvalue and one zero inside $C_p$. 

\begin{theorem}\label{thm-smallcont1}
Let us assume that inside the contour $C_p$ there are the points $\frac{n \pi}{\ell_i}$ and $k_{in} = \sqrt{\lambda_{in}}$ for a given $i$. Then $\lambda_{in} = k_{in}^2$ behaves asymptotically as
\begin{multline*}
  \lambda_{in} = \left(\frac{n\pi}{\ell_i}\right)^2 +\frac{2}{\ell_i}[a_i-\mathrm{Tr\,}H_i-(-1)^n2\mathrm{Re\,}H_{12i}] +
\\
 + \frac{2}{n\pi}\sum_{\stackrel{j = 1}{j\ne i}}^d\left[\cot{\frac{n\pi\ell_j}{\ell_i}}(|H_{11ij}|^2+|H_{12ij}|^2+|H_{21ij}|^2+|H_{22ij}|^2)+\right.
\\
 \left.+\frac{1}{\sin{\frac{n\pi\ell_j}{\ell_i}}}2\mathrm{Re\,}(H_{11ij}\bar{H}_{12ij}+H_{22ij}\bar{H}_{21ij})+\right.
\\
 \left.\frac{(-1)^n}{\sin{\frac{n\pi \ell_j}{\ell_i}}} 2 \mathrm{Re\,}(H_{12ij}\bar{H}_{21ij}+H_{11ij}\bar{H}_{22ij})+\right.
\\
 \left.(-1)^n\cot{\frac{n\pi\ell_j}{\ell_i}}2\mathrm{Re\,}(H_{11ij}\bar{H}_{21ij}+H_{22ij}\bar{H}_{12ij})\right]+\bo \left(\frac{1}{n^2}\right)\,.
\end{multline*}
\end{theorem}
\begin{proof}
We use the integral
$$
  k_{in}^2 - \left(\frac{n\pi}{\ell_i}\right)^2 = -\frac{1}{2\pi i} \oint_{C_p} \ln\frac{\varphi(k)}{\prod_{j = 1}^d (-k\sin{(k\ell_j)})} 2k\,\mathrm{d}k\,.
$$
A straightforward computation using equation~\eqref{eq:lnphi/prod} and Lemma~\ref{lem:integrals} leads to the result.
\end{proof}

When the number of square roots of the eigenvalues (and zeros of the product) is larger than one, we sum over the eigenvalues.
\begin{theorem}\label{thm-smallcont2}
Let us assume that inside the contour $C_p$ there are the points $\frac{n_i \pi}{\ell_i}$ and $k_{in} = \sqrt{\lambda_{in_i}}$ for $i$ from the index set $I$. Then $\sum_{i\in I}\lambda_{in_i}$ behaves asymptotically as
\begin{multline*}
  \sum_{i\in I}\lambda_{in_i} = \sum_{i\in I}\left(\frac{n_i\pi}{\ell_i}\right)^2 +\sum_{i\in I}\frac{2}{\ell_i}[a_i-\mathrm{Tr\,}H_i-(-1)^{n_i} 2\mathrm{Re\,}H_{12i}] +
\\
 + \sum_{i\in I}\frac{2}{n_i\pi}\sum_{\stackrel{j = 1}{n_i \ell_j \ne n_j\ell_i}}^d\left[\cot{\frac{n_i\pi\ell_j}{\ell_i}}(|H_{11ij}|^2+|H_{12ij}|^2+|H_{21ij}|^2+|H_{22ij}|^2)+\right.
\\
 \left.+\frac{1}{\sin{\frac{n_i\pi\ell_j}{\ell_i}}}2\mathrm{Re\,}(H_{11ij}\bar{H}_{12ij}+H_{22ij}\bar{H}_{21ij})+\right.
\\
 \left.\frac{(-1)^{n_i}}{\sin{\frac{n_i\pi \ell_j}{\ell_i}}} 2 \mathrm{Re\,}(H_{12ij}\bar{H}_{21ij}+H_{11ij}\bar{H}_{22ij})+\right.
\\
 \left.(-1)^{n_i}\cot{\frac{n_i\pi\ell_j}{\ell_i}}2\mathrm{Re\,}(H_{11ij}\bar{H}_{21ij}+H_{22ij}\bar{H}_{12ij})\right]+\bo \left(\max_{i\in I}\frac{1}{n_i^2}\right)\,.
\end{multline*}
\end{theorem}
\begin{proof}
Again, we obtain a similar integral as in the previous lemma
$$
  \sum_{i\in I} \left[k_{in}^2 - \left(\frac{n\pi}{\ell_i}\right)^2 \right]= -\frac{1}{2\pi i} \oint_{C_p} \ln\frac{\varphi(k)}{\prod_{j = 1}^d (-k\sin{(k\ell_j)})} 2k\,\mathrm{d}k\,.
$$
If there are no common zeros of the different sine functions, we may apply the same argument as in the previous lemma and obtain the sum of the right-hand side of the previous theorem. If there is a multiple zero of a sine function (i.e. $n_i \ell_j = n_j \ell_i$ for any $i,j$ so that $\fr{n_i\pi}{\ell_i}$ lies inside the contour $C_p$), we may apply Lemma~\ref{lem:integrals} g), i), and k) to show that the contribution of this zero to the third term on the \emph{rhs} is of order $\bo \left(\frac{1}{n_i^2}\right)$.
\end{proof}

Combining the previous two theorems together yields the following corollary
\begin{corollary}\label{cor:abscon}
The sum 
$$
  \sum_{i=1}^d \sum_{n=0}^\infty \left[\lambda_{in}(q)-\lambda_{in}(0)-\frac{2a_i}{\ell_i}\right]
$$
is absolutely convergent, where $\lambda_{in}(q)$ and and $\lambda_{in}(0)$ denote the eigenvalues for the potential $q$ and for the null potential,
respectively.
\end{corollary}
\begin{proof}
Subtracting the right-hand side of the formul\ae\ in Theorems~\ref{thm-smallcont1} and~\ref{thm-smallcont2} one obtains the terms of the sum.
(Note that the term by $\frac{1}{n}$ depends only on the matrix $H$ and not on the potential.) Hence the sum $\sum_{i\in I}\lambda_{in_i}(q)-\lambda_{in_i}(0)
\frac{2a_i}{\ell_i}$ is of order $\bo \left(\max_{i\in I}\frac{1}{n_i^2}\right)$ and the sum of these sums is absolutely convergent.
\end{proof}

Finally, we can prove the main result.

\begin{proof}[Proof of Theorem~\ref{thm:main}]
We integrate around the contours $\Gamma_N$ in the ``allowed regions'' with $N$ going to infinity. For sufficiently large $N$, there are
$d+\sum_{i = 1}^d \left\lfloor \frac{N\ell_i}{\pi}\right\rfloor$ eigenvalues of $\mathcal{H}$ with square roots smaller than $N$ (here
$\lfloor \cdot\rfloor$ denotes the floor function, that is, the largest integer not larger than its argument).
The number of $k_n$ with the same property in the $k$-plane is double. We obtain
\begin{equation}\label{int1}
  2 \sum_{i=1}^d\sum_{n =0}^{\lfloor\frac{N\ell_i}{\pi} \rfloor} [\lambda_{in}(q)-\lambda_{in}(0)] = -\frac{1}{2\pi i} \oint_{\Gamma_N} \ln\frac{\varphi(k)}{\varphi_0(k)} 2k\,\mathrm{d}k\,.
\end{equation}
We can evaluate the integral with the use of equation~\eqref{eq:lnphiphi0} and Lemma~\ref{lem:integrals}, we find after dividing the equation by 2
\begin{equation}
  \sum_{i=1}^d\sum_{n =0}^{\lfloor\frac{N\ell_i}{\pi} \rfloor} [\lambda_{in}(q)-\lambda_{in}(0)] = \sum_{i = 1}^d \frac{a_i}{\ell_i}\left(1+2\left\lfloor\frac{N\ell_i}{\pi}\right\rfloor\right)+\sum_{i = 1}^d \left(b_i-\frac{1}{2}a_i^2\right)+\bo \left(\frac{1}{N}\right)\,.\label{eq:mainuptoN}
\end{equation}
We have used the sums
$$
  \sum_{n = 1}^M \frac{1}{n^2} = \frac{\pi^2}{6}+\bo \left(\frac{1}{M}\right)\,,\quad \sum_{n = 1}^M \frac{(-1)^n}{n^2} = -\frac{\pi^2}{12}+\bo \left(\frac{1}{M}\right)\,.
$$
Subtracting $\sum_{i = 1}^d \frac{2a_i}{\ell_i}\left(1+\left\lfloor\frac{N\ell_i}{\pi}\right\rfloor\right)$ from both sides of \eqref{eq:mainuptoN}, using
$$
  b_i -\frac{1}{2}a_i^2 = \frac{1}{4}[q_i(\ell_i)+q_i(0)],
$$
and sending $N$ to infinity we find the sought result. The contribution of the term $\so\left(\fr{1}{k^2}\right)$ in~\eqref{int1} resulting from
the logarithm expansion~\eqref{eq:lnphiphi0} goes to zero as $N\to \infty$, because the length of the contour is of order $N$ and the value of the function on it is $\so\left(\fr{1}{N^2}\right)
\times N$.
\end{proof}

\appendix
\section{Proof of Lemma~\ref{lem:cs} }\label{sec:appendixa}
First, we prove a version of Lemma~\ref{lem:cs} containing more terms than those used in Section~\ref{sec:seceq} -- this is already partially given in \cite{Yur}. 
\begin{lemma}\label{lem:cs-app}
The functions $c_{j}$ and $s_{j}$ defined above satisfy
\begin{multline*}
  c_j(x,k) = \cos{(kx)}+\frac{\sin{(kx)}}{k} \frac{1}{2}\int_0^x q_j(t)\,\mathrm{d}t+
\\
+ \frac{\cos{(kx)}}{k^2}\left[\frac{1}{4}(q_j(x)-q_j(0))-\frac{1}{8}\left(\int_0^x q_j(t)\,\mathrm{d}t\right)^2\right]+\so\left(\frac{\mathrm{e}^{|\mathrm{Im\,k}|x}}{k^2}\right)\,,
\end{multline*}
\begin{multline*}
  c_j'(x,k) = -k\sin{(kx)}+\cos{(kx)}\frac{1}{2}\int_0^x q_j(t)\,\,\mathrm{d}t+
\\
  + \frac{\sin{(kx)}}{k}\left[\frac{1}{4}(q_j(x)+q_j(0))+\frac{1}{8}\left(\int_0^x q_j(t)\,\mathrm{d}t\right)^2\right]+\so\left(\frac{\mathrm{e}^{|\mathrm{Im\,k}|x}}{k}\right)\,,
\end{multline*}
\begin{multline*}
  s_j(x,k) = \frac{\sin{(kx)}}{k}-\frac{\cos{(kx)}}{k^2}\frac{1}{2}\int_0^x q_j(t)\,\,\mathrm{d}t+
\\
  + \frac{\sin{(kx)}}{k^3}\left[\frac{1}{4}(q_j(x)+q_j(0))-\frac{1}{8}\left(\int_0^x q_j(t)\,\mathrm{d}t\right)^2\right]+\so\left(\frac{\mathrm{e}^{|\mathrm{Im\,k}|x}}{k^3}\right)\,,
\end{multline*}
\begin{multline*}
  s_j'(x,k) =\cos{(kx)}+\frac{\sin{(kx)}}{k}\frac{1}{2}\int_0^x q_j(t)\,\,\mathrm{d}t-
\\
  - \frac{\cos{(kx)}}{k^2}\left[\frac{1}{4}(q_j(x)-q_j(0))+\frac{1}{8}\left(\int_0^x q_j(t)\,\mathrm{d}t\right)^2\right]+\so\left(\frac{\mathrm{e}^{|\mathrm{Im\,k}|x}}{k^2}\right)\,.
\end{multline*}
\end{lemma}

\begin{proof}[Proof]
For the sake of simplicity we omit the subscript $j$. Repeatedly substituting $c_j$ into its defining formula we get
\begin{multline*}
 c(x,k) = \cos{(kx)}+\int_0^x \frac{\sin{(k(x-t))}}{k}\cos{(kt)}\,q(t)\,\mathrm{d}t +
\\
 + \int_0^x\frac{\sin{(k(x-t))}}{k}q(t)\int_0^t\frac{\sin{(k(t-s))}}{k}q(s)\cos{(ks)}\,\mathrm{d}s\mathrm{d}t + \so\left(\frac{\mathrm{e}^{|\mathrm{Im\,}k|x}}{k^2}\right)\,.
\end{multline*}
Using the trogonometric formula 
$$
 \sin{(\alpha-\beta)}\cos{\beta} = \frac{1}{2}[\sin{\alpha}+\sin{(\alpha-2\beta)}]
$$
we obtain
\begin{multline*}
 c(x,k) = \cos{(kx)}+\frac{1}{2k}\int_0^x [\sin{(kx)}+\sin{(k(x-2t))}]\,q(t)\,\mathrm{d}t +
\\
 + \int_0^x\int_0^t\frac{\sin{(k(x-t))}}{k}q(t)\frac{1}{2k}[\sin{(kt)}+\sin{(k(t-2s))}]q(s)\,\mathrm{d}s\mathrm{d}t + \so\left(\frac{\mathrm{e}^{|\mathrm{Im\,}k|x}}{k^2}\right)\,.
\end{multline*}
Finally, using integration by parts we have
\begin{multline*}
  \int_0^x\sin{(k(x-2t))}q(t)\,\mathrm{d}t = \int_0^x q(t)\frac{\partial}{\partial t}\frac{\cos{(k(x-2t))}}{2k}\,\mathrm{d}t =
\\
 = \frac{1}{2k}[q(x)-q(0)]\cos{(kx)} -\frac{1}{2k}\int_0^x\cos{(k(x-2t))}\frac{\partial q(t)}{\partial t}\,\mathrm{d}t =  
\\
 =\frac{1}{2k}[q(x)-q(0)]\cos{(kx)}+ \so\left(\frac{\mathrm{e}^{|\mathrm{Im\,}k|x}}{k}\right)\,,
\end{multline*}
where we have used the fact that $q\in W^{1,1}(e)$. Using this we can write
\begin{multline*}
  c(x,k) = \cos{(kx)}+\frac{\sin{(kx)}}{k}\frac{1}{2}\int_{0}^x q(t)\,\mathrm{d}t+
\\
 +\frac{1}{4k^2}[q(x)-q(0)]\cos{(kx)} + \frac{1}{4k^2}\int_0^x\int_0^t q(t)q(s)[\cos{(k(x-2t))}-\cos{(kx)}]\,\mathrm{d}s\mathrm{d}t+
\\
 + \frac{1}{2k^2}\int_0^x q(t)\sin{(k(x-t))}\int_0^t\sin{(k(t-2s))}q(s)\,\mathrm{d}s\mathrm{d}t + \so\left(\frac{\mathrm{e}^{|\mathrm{Im\,}k|x}}{k^2}\right)\,.
\end{multline*}
By similar arguments as before (with the use of integration by parts) the term in the last line and the term
$\frac{1}{4k^2}\int_0^x\int_0^t q(t)q(s)[\cos{(k(x-2t))}\,\mathrm{d}s\mathrm{d}t$ are of order $\so\left(\frac{\mathrm{e}^{|\mathrm{Im\,}k|x}}{k^2}\right)$.
Finally, since
\begin{multline*}
  \frac{\cos{(kx)}}{4k^2}\int_0^x q(t)\int_0^t q(s)\,\mathrm{d}s\mathrm{d}t = \frac{\cos{(kx)}}{8k^2}\int_0^x \int_0^x q(t) q(s)\,\mathrm{d}s\mathrm{d}t
 = \frac{\cos{(kx)}}{8k^2}\left(\int_0^x q(t)\,\mathrm{d}t\right)^2\,,
\end{multline*}
we obtain the formula for $c(x,k)$.

The formul\ae\ for the function $s(x,k)$ and the corresponding derivatives can be derived in a similar way. For $c'$ we have
\begin{multline*}
  c'(x,k) = -k\sin{(kx)}+\int_0^x \cos{(k(x-t))}q(t)\cos{(kt)}\,\mathrm{d}t+
\\
 + \int_0^x \cos{(k(x-t))}q(t)\int_0^t\frac{\sin{(k(t-s))}}{k}q(s)\cos{(ks)}\,\mathrm{d}s\mathrm{d}t+ \so\left(\frac{\mathrm{e}^{|\mathrm{Im\,}k|x}}{k}\right) =
\\
 = -k \sin{(kx)}+\int_0^x q(t)\frac{1}{2}[\cos{(kx)}+\cos{(k(x-2t))}]\,\mathrm{d}t + 
\\
 +\int_0^x \cos{(k(x-t))}q(t)\int_0^t \frac{1}{2k}q(s)[\sin{(kt)}+\sin{(k(t-2s))}]\,\mathrm{d}s\mathrm{d}t+ \so\left(\frac{\mathrm{e}^{|\mathrm{Im\,}k|x}}{k}\right)
\end{multline*}
For the different particular terms we get
\begin{multline*}
  \frac{1}{2}\int_0^x q(t)\cos{(k(x-2t))}\,\mathrm{d}t = \frac{1}{2} \int_0^x q(t) \frac{\partial \sin{(k(x-2t))}}{\partial t}\left(-\frac{1}{2k}\right)\,\mathrm{d}t = 
\\ 
  =\frac{1}{4k}\int_0^x \frac{\partial q(t)}{\partial t} \sin{(k(x-2t))} \,\mathrm{d}t + \frac{1}{4k}[q(x)+q(0)]\sin{(kx)} = 
\\ 
 = \frac{1}{4k}[q(x)+q(0)]\sin{(kx)} + \so\left(\frac{\mathrm{e}^{|\mathrm{Im\,}k|x}}{k}\right)\,.
\end{multline*}
\begin{multline*}
  \int_0^x \cos{(k(x-t))}q(t)\int_0^t \frac{1}{2k}q(s)\sin{(kt)}\,\mathrm{d}s\mathrm{d}t =
\\ 
 =  \frac{1}{4k}\int_0^x q(t)[\sin{(kx)}-\sin{(k(x-2t))}]\int_0^t q(s)\,\mathrm{d}t =
\\ 
 = \frac{\sin{(kx)}}{8k}\left(\int_0^x q(t)\,\mathrm{d}t\right)^2  + \so\left(\frac{\mathrm{e}^{|\mathrm{Im\,}k|x}}{k}\right)\,.
\end{multline*}
$$
  \int_0^x\cos{(k(x-t))}q(t)\int_0^t \frac{1}{2k}q(s)\sin{(k(t-2s))}\,\mathrm{d}s\mathrm{d}t =  \so\left(\frac{\mathrm{e}^{|\mathrm{Im\,}k|x}}{k}\right)\,.
$$
We also briefly show the derivation of formul\ae\ for $s$ and $s'$.
\begin{multline*}
  s(x,k) = \frac{\sin{(kx)}}{k} + \int_0^x \frac{\sin{(k(x-t))}}{k^2}q(t)\sin{(kt)}\,\mathrm{d}t +
\\
 + \int_0^x  \frac{\sin{(k(x-t))}}{k^3}q(t)\int_0^t \sin{(k(t-s))}q(s)\sin{(ks)}\,\mathrm{d}s\mathrm{d}t + \so\left(\frac{\mathrm{e}^{|\mathrm{Im\,}k|x}}{k^3}\right) = 
\\
 =\frac{\sin{(kx)}}{k}-\frac{\cos{(kx)}}{k^2}\frac{1}{2}\int_{0}^x q(t)\,\mathrm{d}t + \frac{1}{4k^3}\sin{(kx)}[q(x)+q(0)]-
\\
 -\frac{1}{4k^3}\int_0^x[\sin{(kx)}+\sin{(k(x-2t))}]q(t)\int_0^t q(s)\,\mathrm{d}s\mathrm{d}t+ 
\\
 +\frac{1}{2k^3}\int_0^x\sin{(k(x-t))}q(t)\int_0^t \cos{(k(t-2s))}q(s)\,\mathrm{d}s\mathrm{d}t+\so\left(\frac{\mathrm{e}^{|\mathrm{Im\,}k|x}}{k^3}\right)\,.
\end{multline*}
\begin{multline*}
  s'(x,k) = \cos{(kx)}+\frac{1}{k}\int_0^x\cos{(k(x-t))}\sin{(kt)}q(t)\,\mathrm{d}t +
\\
 +\frac{1}{k^2}\int_0^x \cos{(k(x-t))}q(t)\int_0^t\sin{(k(t-s))}\sin{(ks)}q(s)\,\mathrm{d}s\mathrm{d}t + \so\left(\frac{\mathrm{e}^{|\mathrm{Im\,}k|x}}{k^2}\right) =
\\
 = \cos{(kx)}+\frac{1}{2k}\int_0^x [\sin{(kx)}-\sin{(k(x-2t))}]q(t)\,\mathrm{d}t+
\\
 +\frac{1}{2k^2}\int_0^x \cos{(k(x-t))}q(t)\int_0^t[\cos{(k(t-2s))}-\cos{(kt)}]q(s)\,\mathrm{d}s\mathrm{d}t = 
\\
 + \so\left(\frac{\mathrm{e}^{|\mathrm{Im\,}k|x}}{k^2}\right) = \cos{(kx)}+\frac{\sin{(kx)}}{k}\frac{1}{2}\int_0^x q(t)\,\mathrm{d}t - \frac{\cos{(kx)}}{4k^2}[q(x)-q(0)]-
\\
 -\frac{1}{4k^2}\int_0^x\int_0^t \cos{(k(x-2t))}q(t)q(s)\,\mathrm{d}s\mathrm{d}t-
\\
 -\frac{1}{4k^2}\int_0^x \cos{(kx)}q(t)\int_0^t q(s)\mathrm{d}s\mathrm{d}t+\so\left(\frac{\mathrm{e}^{|\mathrm{Im\,}k|x}}{k^2}\right)\,.
\end{multline*}

\end{proof}

\section{Proof of Theorem~\ref{thm:numberofzeros}}\label{sec:appendixb}

\begin{lemma}\label{lem:smallo}
On the contour $\Gamma_N$ defined in Section~\ref{sec:main} with large enough $N$ satisfying
\[
N \not \in {\ds \bigcup_{n\in\mathbb{Z}} } \left(\frac{n\pi}{\ell_j}-\frac{\varepsilon}{\ell_j},\frac{n\pi}{\ell_j}+\frac{\varepsilon}{\ell_j}\right),
\]
it holds 
$$
  \frac{\mathrm{e}^{|\mathrm{Im\,}k|\ell_j}}{|\sin{(k\ell_j)}|} \leq K_{\varepsilon}\,,
$$
where the constant $K_{\varepsilon}$ depends only on $\varepsilon$.
\end{lemma}
\begin{proof}
The proof will be similar to the proof of \cite[Lemma 2.4]{Yan14}.
We will first prove the inequality for the right edge of the square $\Gamma_N$, i.e. for $k = N+i \tau$, $\tau \in (-N,N)$.
We know that there exist such $C_\varepsilon>0$ that $|\sin{(k\ell_j)}|> C_{\varepsilon}$. We have
\begin{multline*}
  |\sin{(k\ell_j)}| = |\sin{(N\ell_j)}\cos{(i\tau\ell_j)}+\cos{(N\ell_j)}\sin{(i\tau\ell_j)}| = 
\\
 =\frac{1}{2}\sqrt{|\sin{(N\ell_j)}(\mathrm{e}^{-\tau \ell_j}+\mathrm{e}^{\tau \ell_j})|^2+|\cos{(N\ell_j)}(\mathrm{e}^{-\tau \ell_j}-\mathrm{e}^{\tau \ell_j})|^2} \geq 
\\
 \geq\frac{1}{2}|\sin{(N\ell_j)}||\mathrm{e}^{-\tau \ell_j}+\mathrm{e}^{\tau \ell_j}|> \frac{1}{2}C_\varepsilon \mathrm{e}^{|\tau| \ell_j}
\end{multline*}
and hence
$$
  \frac{\mathrm{e}^{|\mathrm{Im\,}k|\ell_j}}{|\sin{(k\ell_j)}|} \leq \frac{2\mathrm{e}^{|\mathrm{Im\,}k|\ell_j}}{C_\varepsilon \mathrm{e}^{|\mathrm{Im\,}k|\ell_j}} = \frac{2}{C_\varepsilon}\,.
$$
For the upper edge of the square $k = \sigma+iN$, $\sigma\in(-N,N)$ we have for sufficiently large $N$
$$
  |\sin{(k\ell_j)}| = \frac{1}{2}|\mathrm{e}^{-N\ell_j+i \sigma\ell_j}-\mathrm{e}^{N\ell_j-i \sigma\ell_j}|\geq \frac{1}{2}(\mathrm{e}^{N\ell_j}-\mathrm{e}^{-N\ell_j})
$$
and hence for $N$ large enough
$$
  \frac{\mathrm{e}^{|\mathrm{Im\,}k|\ell_j}}{|\sin{(k\ell_j)}|} \leq \frac{2\mathrm{e}^{N\ell_j}}{\mathrm{e}^{N\ell_j}-\mathrm{e}^{-N\ell_j}} \leq 4\,.
$$
We have chosen $N$ such that $\mathrm{e}^{-2N\ell_j}<\frac{1}{2}$. The proof for the other edges of the square $\Gamma_N$ is similar.
\end{proof}

For the sake of completeness we present the symmetric version of Rouch\'e's theorem (for the proof see e.g. \cite[p. 156]{Est} or \cite[p. 265]{Bur}).
\begin{theorem}\label{thm:rouche}
Let $f$ and $g$ be holomorphic functions in the bounded subset $V$ of $\mathbb{C}$ and continuous at its closure $\bar{V}$. Let us assume that on the boundary $\partial V$ of $V$ the following relation holds
$$
  |f-g|<|f|+|g|\,.
$$ 
Then functions $f$ and $g$ have the same (finite) number of zeros in $V$.
\end{theorem}

Now we can proceed with the proof of Theorem~\ref{thm:numberofzeros}.
\begin{proof}[Proof of Theorem~\ref{thm:numberofzeros}]
Since we assume that $N\not \in \cup_{n\in \mathbb{N}_0} \left(\fr{n\pi}{\ell_i}-\fr{\varepsilon}{\ell_i}, \fr{n\pi}{\ell_i}+\fr{\varepsilon}{\ell_i}\right)$ for each $i$, we have $|\sin{(N\ell_i)}|> C_\varepsilon > 0$ with $C_\varepsilon$ depending only on $\varepsilon$. We use the Rouch\'e's theorem with $f= \varphi(k)$ and $g = \prod_{i = 1}^d (-k\sin{(k\ell_i)})$. 

Using $\sin(k\ell_j) = \bo (\mathrm{e}^{|\mathrm{Im\,}k|\ell_j})$ and a similar relation for the cosinus one can find that the second and further terms in $\varphi(k)$
belong to $\bo (|k|^{d-1}\mathrm{e}^{|\mathrm{Im\,}k|\sum_{i=1}^d\ell_i})$. On the contour $\Gamma_N$ we hence have
\begin{eqnarray*}
  |f|+|g| = 2 |k|^d \prod_{i=1}^d |\sin{(k\ell_i)}|+\bo (|k|^{d-1}\mathrm{e}^{|\mathrm{Im\,}k|\sum_{i=1}^d\ell_i})\,,\\
  |f-g| \leq \bo (|k|^{d-1}\mathrm{e}^{|\mathrm{Im\,}k|\sum_{i=1}^d\ell_i})\,.
\end{eqnarray*}
Using Lemma~\ref{lem:smallo} we obtain
$$
  |f|+|g|-|f-g| = |k|^d\prod_{i = 1}^d|\sin{(k\ell_i)}|\left(2+\frac{1}{|k|}\bo \left(\frac{\mathrm{e}^{|\mathrm{Im\,}k|\sum_{j=1}^d\ell_j}}{\prod_{o = 1}^d|\sin{(k\ell_o)}|}\right)\right) > |k|^d C_\varepsilon^d>0
$$
for $|k|$ large enough and hence the inequality in Theorem~\ref{thm:rouche} is satisfied which completes the proof.
\end{proof}

\section{Complex integration lemma}\label{sec:appendixc}

\begin{lemma}\label{lem:integrals}
Let us assume a counterclockwise contour $\gamma_n$ which encircles $\fr{n\pi}{\ell_i}$ once and does not encircle any other zeros of $\sin{(k\ell_i)}$. Then
\[
\begin{array}{ll}
{\rm a)}\ \fr{1}{2\pi i}\doint_{\gamma_n} \cot{(k\ell_i)} \,\mathrm{d}k = \frac{1}{\ell_i}, n\in \mathbb{Z} &
{\rm b)}\ \fr{1}{2\pi i}\doint_{\gamma_n} \frac{1}{\sin{(k\ell_i)}} \,\mathrm{d}k = \frac{(-1)^n}{\ell_i}, n\in \mathbb{Z}\eqskip
{\rm c)}\ \fr{1}{2\pi i}\doint_{\gamma_n} \frac{1}{k}\cot{(k\ell_i)} \,\mathrm{d}k = \frac{1}{n \pi}, n \in \mathbb{Z}\backslash\{0\}&
{\rm d)}\ \fr{1}{2\pi i}\doint_{\gamma_0} \frac{1}{k}\cot{(k\ell_i)} \,\mathrm{d}k = 0\eqskip
{\rm e)}\ \fr{1}{2\pi i}\doint_{\gamma_n} \frac{1}{k\sin{(k\ell_i)}} \,\mathrm{d}k = \frac{(-1)^n}{n \pi}, n \in \mathbb{Z}\backslash\{0\}&
{\rm f)}\ \fr{1}{2\pi i}\doint_{\gamma_0} \frac{1}{k\sin{(k\ell_i)}} \,\mathrm{d}k = 0\eqskip
{\rm g)}\ \fr{1}{2\pi i}\doint_{\gamma_n} \frac{1}{k}\cot^2{(k\ell_i)} \,\mathrm{d}k = -\frac{1}{n^2 \pi^2}, n \in \mathbb{Z}\backslash\{0\}&
{\rm h)}\ \fr{1}{2\pi i}\doint_{\gamma_0} \frac{1}{k}\cot^2{(k\ell_i)} \,\mathrm{d}k = -\frac{2}{3}\eqskip
{\rm i)}\ \fr{1}{2\pi i}\doint_{\gamma_n} \frac{\cot{(k\ell_i)}}{k\sin{(k\ell_i)}} \,\mathrm{d}k = -\frac{(-1)^n}{n^2 \pi^2}, n \in \mathbb{Z}\backslash\{0\}&
{\rm j)}\ \fr{1}{2\pi i}\doint_{\gamma_0} \frac{\cot{(k\ell_i)}}{k\sin{(k\ell_i)}} \,\mathrm{d}k = -\frac{1}{6}\eqskip
{\rm k)}\ \fr{1}{2\pi i}\doint_{\gamma_n} \frac{1}{k\sin^2{(k\ell_i)}} \,\mathrm{d}k = -\frac{1}{n^2 \pi^2}, n \in \mathbb{Z}\backslash\{0\}&
{\rm l)}\ \fr{1}{2\pi i}\doint_{\gamma_0} \frac{1}{k\sin^2{(k\ell_i)}} \,\mathrm{d}k = \frac{1}{3}\eqskip
{\rm l)}\ \fr{1}{2\pi i}\doint_{\gamma_0} \frac{1}{k} \,\mathrm{d}k = 1
\end{array}
\]
\end{lemma}
\begin{proof}
The lemma can be proven by standard complex analysis techniques, i.e. the residue theorem, see e.g. \cite{Bur}.
\end{proof}

\section*{Acknowledgements}
P.F. was partially supported by the Funda{\c{c}}{\~a}o para a Ci{\^e}ncia e a Tecnologia, Portugal, through project UIDB/00208/2020.
J.L. was supported by the project ``International mobilities for research activities of the University of Hradec Kr\'alov\'e'' CZ.02.2.69/0.0/0.0/16{\_}027/0008487.
J.L. thanks the University of Lisbon for its hospitality during his stay in Lisbon. The authors are grateful to the reviewer for the suggestions which
helped to improve the manuscript.  Data sharing not applicable to this article as no datasets were generated or analysed during the current study. This is a preprint of an article published in Anal. Math. Phys. The final authenticated version is available online at: \texttt{https://doi.org/10.1007/s13324-021-00487-3}.

\def\cprime{$'$}

\end{document}